\documentclass[11pt]{article}
\usepackage{enumerate, amsthm, amsmath, amssymb, hyperref}
\newcommand{\logrank}{\widetilde{\textrm{Logrank}}}
\theoremstyle{plain}
\newtheorem{theorem}{Theorem}[section]
\newtheorem{lemma}[theorem]{Lemma}
\newtheorem{corollary}[theorem]{Corollary}
\newtheorem{proposition}[theorem]{Proposition}

\newtheorem{conjecture}[theorem]{Conjecture}

\theoremstyle{definition}
\newtheorem{definition}[theorem]{Definition}

\theoremstyle{remark}
\newtheorem{remark}[theorem]{Remark}

\newcommand{\bool}{\{0, 1\}}

\def\squareforqed{\hbox{\rlap{$\sqcap$}$\sqcup$}}
\def\qed{\ifmmode\squareforqed\else{\unskip\nobreak\hfil
\penalty50\hskip1em\null\nobreak\hfil\squareforqed
\parfillskip=0pt\finalhyphendemerits=0\endgraf}\fi}

\newenvironment{proofof}[1]{\begin{trivlist}%
\item[]{\flushleft\em Proof of #1. }}
{\qed\end{trivlist}}

\def\disc{\textrm{disc}}
\def\g{\mathcal{G}}
\def\>{\ensuremath{\rangle}}
\def\<{\ensuremath{\langle}}
\def\disj{\ensuremath{\textsc{DISJ}}}
\def\dom{\ensuremath{\textrm{dom}}}
\def\disjP{\ensuremath{\textsc{DISJ}^{\le 1}}}
\newcommand{\appd}{\widetilde{\textrm{deg}}}
\newcommand{\IP}{\textsc{IP}}
\newcommand{\trace}{\mathrm{tr}}

\newcommand{\comments}[1]{}

\newcommand{\defeq}{\stackrel{\mathrm{def}}{=}}
\newcommand{\ham}{\textsc{Ham}}

\newcommand{\sepAuthor}{.3in}
\newcommand{\sepAbstract}{.3in}
\newcommand{\skipKeywords}{30pt}
\newcommand{\sepTitle}{2ex}

\long\def\mytitlepage#1#2#3#4{
        \thispagestyle{empty}
    \vspace*{\sepTitle}
        \begin{center}
        {\Large\bf #1}

        \vspace{\sepAuthor}
        #2\\
        \medskip

        \vspace{\sepAbstract}
        {\Large Abstract}
        \end{center}

        \noindent{#3}
        \vskip\skipKeywords

        \noindent{#4}
        \clearpage
        }
\begin{document}
\mytitlepage{Quantum communication complexity of block-composed
functions }
{{\large Yaoyun Shi\quad and \quad Yufan Zhu}\\
\vspace{1ex}
Department of Electrical and Computer Engineering\\
The University of Michigan\\
Ann Arbor, MI 48109-2122, USA\\
E-mail: \{shiyy$|$yufanzhu\}@eecs.umich.edu\\
\vspace{2ex}
{\today}
}
{
A major open problem in communication complexity is whether 
or not quantum protocols can be exponentially more efficient than 
classical ones for computing a {\em total} Boolean function in the 
two-party interactive model. The answer appears to be ``No''. 
In 2002, Razborov proved this conjecture for so far 
the most general class of functions 
$F(x, y)=f_n(x_1\cdot  y_1, x_2\cdot  y_2, ..., x_n\cdot y_n)$, 
where $f_n$ is a {\em symmetric} Boolean function on $n$ Boolean inputs, 
and $x_i$, $y_i$ are the $i$'th bit of $x$ and $y$, respectively. 
His elegant proof critically depends on the symmetry of $f_n$.

We develop a lower-bound method that does not require symmetry
and prove the conjecture for a broader class of functions. 
Each of those functions $F(x, y)$ is the 
``block-composition'' of a ``building block'' 
$g_k : \{0, 1\}^k \times \{0, 1\}^k \rightarrow \{0, 1\}$, 
and an $f_n : \{0, 1\}^n \rightarrow \{0, 1\}$, 
such that $F(x, y) = f_n( g_k(x_1, y_1), g_k(x_2, y_2), ..., g_k(x_n, y_n) )$, 
where $x_i$ and $y_i$ are the $i$'th $k$-bit block of $x, y\in\{0, 1\}^{nk}$, 
respectively. 

We show that as long as $g_k$ itself is ``hard'' enough,
its block-composition with an {\em arbitrary} $f_n$ has polynomially related 
quantum and classical communication complexities. Our approach 
gives an alternative proof for Razborov's result (albeit with a
slightly weaker parameter), and establishes new quantum lower bounds. 
For example, when $g_k$ is the Inner Product function 
with  $k=\Omega(\log n)$, 
the {\em deterministic} communication complexity of its block-composition 
with {\em any} $f_n$ is asymptotically at most the quantum complexity to 
the power of $7$.
}{\noindent{\bf Keywords:} Communication complexity, quantum
information processing, polynomial approximation of Boolean
functions, quantum lower bound. }

\section{Introduction and summary of results}
Communication complexity studies the inherent communication cost for
distributive computations. Let $F: X\times Y\to\bool$ be a function
which two parties Alice, who knows $x\in X$, and Bob, who knows
$y\in Y$, wish to compute. The {\em communication complexity} of $F$
is the minimum amount of information they need to exchange to
compute $F$ on the worst case input. There are several variants of
communication complexity: each of which corresponds to different 
types of interactions allowed and whether or not small error 
probability is allowed. For example, we study the following three
variants in this paper: deterministic (denoted by $D(F)$), 
randomized (denoted by $R(F)$), and quantum (denoted by $Q(F)$). 
In the last two cases, the protocol may err with probability 
$\le 1/3$. Since its introduction by Yao~\cite{Yao:1979:comm} 
in 1979, the study of communication complexity has developed into a
major branch of complexity theory, with a wide range of
applications such as in VLSI design, time-space tradeoff,
derandomization, and circuit complexity. The monograph
by Kushilevitz and Nisan \cite{Kushilevitz:1997:book} 
surveys results up to 1997.

Quantum communication complexity was introduced by 
Yao \cite{Yao:1993:circuit} in 1993, and has been studied extensively
since then. 
A major problem in this area
is to identify problems that have an exponential gap between 
quantum and classical communication complexities, or to prove that 
such a problem does not exist.

Exponential gaps were indeed found for several communication tasks
(\cite{ASTS+98a, Raz99a, Bar-Yossef:2004:com, Gavinsky:2005:comm,
Kerenidis06, Gavinsky06}). However, those tasks are either
sampling, or computing a partially defined function or a relation.
An exponential gap is known for a total Boolean function (checking
equality), but in a restricted model that involves a third party
({\em Simultaneous Message Passing} model) \cite{BCWdW01a}. It remains
open to day if super-polynomial gaps are possible for computing a {\em total
Boolean function} in the more commonly studied model of two-party interactive
communication. This is one of the most significant problems in
quantum communication complexity, and is the question we address
in this article.

It is believed that the answer to the above question is ``No'':
\begin{conjecture}[Log-Equivalence Conjecture]\label{conj:equiv} For any total
Boolean function, the quantum (with shared entanglement) 
and randomized (with shared randomness) communication
complexities are polynomially related in the two-way interactive
model.
\end{conjecture}
Besides the lack of a natural candidate for a super-polynomial gap,
two other intuitions support this conjecture. The first 
relates to 
the approximate version of the well
known Log-Rank Conjecture, which states
that for any $F: X\times Y\to\bool$, $R(F)$
is polynomially related to
$\logrank(F)$, the logarithm of the smallest rank
of a real-valued matrix  $[\tilde F(x, y)]_{x, y}$
approximating $[F(x, y)]_{x\in X, y\in Y}$ entry-wise.
It is known that without sharing entanglement,
the quantum complexity of $F$ is at least $\frac{1}{2}\logrank(F)$.
Thus the Log-Equivalence Conjecture
follows from the Log-Rank Conjecture, unless
there exist exponential gaps between quantum protocols
with or without shared entanglement. The existence of such
gaps is also a fundamental open problem currently under
active investigations.

The second intuition supporting the Log-Equivalence Conjecture
is the fact that the similar conjecture is true for 
the closely related decision tree complexity.
Recall that a decision tree algorithm computes a function
$f_n:\{0, 1\}^n\to\{0, 1\}$ by making queries of the type
``what is the $i$'th bit of the input?''
The decision tree complexity of $f_n$ is the minimum number of 
queries required to compute $f_n$ correctly for any input. 
Making use earlier results of Nisan and Szegedy~\cite{NisanS92} 
and Paturi~\cite{Paturi}, Beals, Buhrman, Cleve, Mosca, and 
de~Wolf~\cite{Beals:2001:QLB} proved that the quantum and
the deterministic decision tree complexities are polynomially 
related. This is in sharp contrast with the exponential quantum 
speedups~\cite{Simon, Shor, Cleve} on {\em partial functions} 
achieved by the quantum algorithms of Simon's and Shor's.

Razborov's work~\cite{Razborov} is a significant progress
for the Log-Equivalence Conjecture. He defined the following 
notion of {\em symmetric predicates}. Let 
$f_n:\{0, 1\}^n\to\{0, 1\}$ be a symmetric function, i.e., 
$f_n(x)$ depends only on the Hamming weight of $x$. A 
function $F:\{0, 1\}^n\times \{0, 1\}^n\to\{0, 1\}$ 
is called a {\em symmetric predicate} if 
$F(x, y)=f(x_1\wedge y_1, x_2\wedge y_2, \cdots, x_n\wedge y_n)$.
The \textsc{Disjointness} function $\disj_n$ is an important symmetric
predicate that has been widely studied:
\begin{equation*} \disj_n(x, y)\defeq\left\{\begin{array}{cc}1&\textrm{
$\exists i,\ x_i=y_i=1$,}\\ 0&\textrm{otherwise.}\end{array}\right.
\end{equation*}
\begin{theorem}[Razborov~\cite{Razborov}] For any symmetric
predicate $F:\{0, 1\}^n\times\{0, 1\}^n\to\{0, 1\}$,
$D(F)=O(\max\{Q(F)^2, Q(F)\log n\})$.
\end{theorem}

Combined with the $O(d\log d)$-bit classical protocol
of Huang et al.~\cite{Hamming} for deciding if $x, y\in\{0, 1\}^n$
has Hamming distance $|x\oplus y|\ge d$, Razborov's lower
bound implies the following.

\begin{proposition}\label{prop:quadGap}For any symmetric predicate $F:\{0, 1\}^n\times
\{0, 1\}^n\to\{0, 1\}$, $R(F)=O((Q(F))^2)$.
\end{proposition}

This bound is tight on $\disj_n$, which admits
the largest known quantum-classical gap 
for total Boolean functions. The class of symmetric
predicates is also the most general class of functions on 
which the Log-Equivalence Conjecture is known to hold.

Notice that Razborov's lower bound method relies on the 
{\em symmetry} of $f_n$. Thus we aim to develop 
lower-bound techniques for an arbitrary $f_n$,
and to derive new quantum lower bounds.
To this end, we consider the following
class of functions.

\begin{definition} Let $k, n\ge 1$ be integers. Given
$f_n:\{0, 1\}^n\to \{0, 1\}$, and 
$g_k:\{0, 1\}^k\times \{0, 1\}^k\to\{0, 1\}$, 
the {\em block-composition} of $f_n$ and $g_k$ is the function
$f_n\Box g_k : \{0, 1\}^{nk}\times \{0, 1\}^{nk}\to\{0, 1\}$ such
that on $x, y\in \{0, 1\}^{nk}$, with $x=x_1x_2\cdots x_n$,
and $y=y_1y_2\cdots y_n$, where $x_i, y_i\in \{0, 1\}^k$,
\[ f_n\Box g_k (x, y) = f_n(g_k(x_1, y_1), g_k(x_2, y_2), \cdots, g_k(x_n, y_n)).\]
\end{definition}

Note that a symmetric predicate based on a symmetric $f_n:\bool^n\to\bool$
is the block composition $f_n\Box \wedge$, where $\wedge$ denotes
the binary AND function.
In our Main Lemma, stated and proved in Section~\ref{sec:mainlemma},
we derive a sufficient condition for 
$Q(f_n\Box g_k)$ to have a strong lower bound. 
An application of this Main Lemma
is the following.

\begin{theorem}[Informally]\label{res:hardg:informal} For any integer $n\ge1$
and any function $f_n: \{0, 1\}^n\to\{0, 1\}$, the block
composition of $f_n$ with a $g_k:\{0, 1\}^k\to\{0, 1\}$
has polynomially related quantum and randomized communication complexities,
if $g_k$ is sufficiently hard.
\end{theorem}

We will define what ``sufficiently hard'' means precisely. Roughly,
it means that $Q(g_k)$ and $R(g_k)$ are polynomially related, and 
some type of discrepancy parameter (Definition \ref{def:spectradis}) on $g_k$ is sufficiently small.
We state below an incarnation of the above theorem.
Let $\textsc{IP}_k:\{0, 1\}^k\times \{0, 1\}^k\to\{0, 1\}$ be the 
widely studied Inner Product function
\[ \textsc{IP}_k(x, y) \defeq \sum_i x_iy_i \mod 2,\quad\forall x, y\in\{0, 1\}^k.\]

\begin{corollary}\label{col:ip}
For any integers $k$ and $n$ with $k\ge 2\log_2 n+5$, and for an {\em arbitrary}
$f_n:\{0, 1\}^n\to\{0, 1\}$,
$D(f_n\Box \textsc{IP}_k)=O((Q(f_n\Box \textsc{IP}_k))^7)$.
\end{corollary}

The above corollary also holds for a random $g_k$ with high probability.
Our technique can also be applied to symmetric predicates,
thus giving an alternative proof to Razborov's result, albeit with a
weaker parameter.

\begin{theorem}\label{thm:symmetric} 
For any symmetric $f_n:\{0, 1\}^n\to\{0, 1\}$,
$R(f_n\Box\wedge)=O((Q(f_n\Box\wedge))^3)$.
\end{theorem}

Our approach is inspired by how the Log-Equivalence result
in decision tree complexity was proved: 
for any $f_n:\bool^n\to\bool$, both
the quantum and the deterministic decision tree complexities
were shown~\cite{NisanS92, Beals:2001:QLB}
to be polynomially related to the 
{\em approximate polynomial degree} (or, approximate degree for short) $\appd(f_n)$,
which is the smallest 
degree of a real polynomial that approximate
$f_n$ to be within $1/3$ on any $0/1$ inputs.
In our Main Lemma, we derive a sufficient condition on 
$n$ and $k$, and $g_k$ such
that $Q(f_n\Box g_k)=\Omega(\appd(f_n))$, for any $f_n$.
The randomized upper bound is obtained
by simulating a decision tree algorithm
for $f_n$, and whenever one input bit of $f_n$ is needed,
the protocol calls a sub-protocol for computing $g_k$
on the corresponding block. Under some hardness assumption
on $g_k$, those upper and lower bounds are polynomially related.

The approach for proving a quantum lower bound 
using an approximate degree lower bound is termed the {\em polynomial method}
in the literature of quantum decision trees.
Razborov's lower bound on $\disj$ can be viewed
an application of the polynomial method as well. This is because,
he showed that if there is a $q$-qubit protocol for
$\disj_n$, then there is a $O(q)$-degree polynomial
approximating $\textsc{OR}_n$. Thus the quantum lower bound
of $\Omega(\sqrt{n})$ follows from the same lower bound
on $\appd(\textsc{OR}_n)$ due to Nisan and Szegedy~\cite{NisanS92}
and Paturi~\cite{Paturi}.
We emphasize this connection of approximating polynomial
and quantum protocol is not obvious at all and it makes use the symmetry
of $\disj$ critically.

We avoid the dependence of Razborov's proof on the symmetry 
property of $f_n$ by taking the {\em dual} approach of 
the polynomial method. We show that from the linear programming
formulation of polynomial approximation, we can obtain
a ``witness'' for $f_n$ requiring a high approximate degree.
This witness is then turned into a ``witness'' for the 
hardness of $f_n\Box g_k$, under certain assumptions.
While the approximate polynomial degree has been used to 
prove lower bounds, and its dual formulation 
has been known to several researchers~\cite{RazborovP, Szegedy},
our application of the dual form appears to be the first 
demonstration of its usefulness in proving new results.
We note that there are several recent works that use
the duality of linear (or semidefinite) programming
in the context of communication complexity
~\cite{LinialS07, Sherstov07, LeeSS08, LinialS08}.
Those applications of duality, however, do not involve
the type of polynomial approximation of Boolean functions
considered here.

Before we proceed to the proofs, we briefly review
some other closely related works. Buhrman and de Wolf
\cite{Buhrman:2000:com} are probably the first to systematically
study the relationship of polynomial representations and
communication complexity. However, their result applies to error-free
quantum protocols, while we consider bounded-error case.
Klauck \cite{Klauck01} proved strong lower bounds for some
symmetric predicates such as MAJORITY based on the properties of
their Fourier coefficients. The same author
formulated a lower bound framework that includes several
known lower bound methods~\cite{Klauck:rectangle}. It would be
interesting to investigate the limitations of our polynomial
method in this framework. 
After preparing this draft, we learned about an independent work
by Sherstov~\cite{Sherstov}, who used a similar approach
to prove similar results. We will compare our work with his
in the concluding section.

\section{Preliminaries}
\subsection{Communication complexities and quantum lower bound 
by approximate trace norm}
Denote the domain of a function by $\dom(\cdot)$.
For a positive integer $n$, denote 
by $\mathcal{F}_n\defeq\{f_n: \bool^n\to\bool\}$,
and by $\mathcal{G}_k\defeq\{ g_k: \dom(g_k)\to\bool,\ \dom(g_k)\subseteq\bool^k\times\bool^k\}$.
For the rest of this article $f_n\in\mathcal{F}_n$
and $g_k\in\mathcal{G}_k$, for some integers $n, k\ge1$.

If $F\in\mathcal{G}_n$ is a total function, we also 
denote by $F$ the $\{0, 1\}^{2^n\times 2^n}$ 
matrix $[F(x, y)]_{x, y\in\bool^n}$.
Consider the computation of $F\in\mathcal{G}_n$ on $(x, y)\in\dom(F)$
when the input $x$ is known to a party Alice and $y$ 
is known to another party Bob. 
Unless $F(x, y)$ trivially depends only on $x$ or $y$, 
Alice and Bob will have to communicate before they could 
determine $F(x, y)$. The worst case cost of communication is called 
the communication complexity
of $F$.

The communication complexity depends on the information
processing power of Alice and Bob, and the requirement
on the accuracy of the outcome of a protocol.
Thus we have various communication
complexities: deterministic (denoted by $D(f)$), randomized
($R_\epsilon(f)$), and quantum $(Q_\epsilon(f))$, when the protocols are restricted to
be deterministic, randomized, and quantum, respectively,
and $\epsilon\in(0, 1/2)$ is a constant that upper-bounds the error
probability of the protocols.
In the randomized and the quantum cases we allow Alice and Bob share unlimited
amount of randomness or quantum entanglement, respectively.
Different choices of $\epsilon$ only result in a change of
a constant factor in the communication complexities, by
a standard application of the Chernoff Bound. 
Thus we may omit the subscripts in $R_\epsilon(F)$
and $Q_\epsilon(F)$ for asymptotic estimations.

A powerful method for proving quantum communication
complexity lower bounds is the following lemma, 
which was obtained by Razborov~\cite{Razborov}, 
extending a lemma of Yao~\cite{Yao:1993:circuit}.
Recall that the trace norm of a matrix $A\in\mathbb{R}^{N\times M}$ is
$\|A\|_{tr}\defeq \textrm{trace} \sqrt{A^TA}=\textrm{trace}\sqrt{AA^T}$.
Let $F$ be a partial Boolean function defined on a subset
$\dom(F)\subseteq X\times Y$. The {\em approximate trace norm} 
of $F$ with error $\epsilon$, $0\le \epsilon<1/2$, is
\[ \| F \|_{\epsilon, \trace}\defeq \ \min\{ \|\tilde F\|_{\trace}
: \tilde F\in\mathbb{R}^{N\times M},\ \forall (x, y)\in \dom(F),
\ |\tilde F(x, y) - F(x, y)|\le \epsilon\}.\]
\begin{lemma}[Razborov-Yao\cite{Razborov, Yao:1993:circuit}]
\label{lm:Razborov-Yao} For any partial Boolean function
$F$ whose domain is a subset of $X\times Y$,
$Q_\epsilon(F)=\Omega(\log \frac{\|F\|_{\epsilon, \trace}}{\sqrt{|X|\cdot|Y|}})$.
\end{lemma}

For matrix $B$, denote by $\|B\|$ its operator norm. Since
for any matrix $A$, $\|A\|_{tr} = \sup_{B, \|B\|=1}
|\textrm{trace}(B^T A)|$,  we have
\[ \|A\|_{tr} \ge \frac{|\textrm{trace}(B^T A)|}{\|B\|},\quad
\forall B\ne0.\] Therefore, in order to prove that $\|A\|_{tr}$ is large, 
we need only to find a $B$ so that $|\textrm{trace}(B^TA)|/\|B\|$ is large. 

\subsection{Approximate polynomial degree}
The study of low degree polynomial approximations of Boolean
function under the $\ell_\infty$ norm  was pioneered by Nisan and
Szegedy \cite{NisanS92} and Paturi \cite{Paturi}, and has
since then been a powerful tool in studying concrete complexities,
including the quantum decision tree complexity (c.f. the 
survey by Buhrman and de Wolf \cite{Buhrman:2002:CMD}).

Let $f\in\mathcal{F}_n$.  A real polynomial
$\tilde f:\mathbb{R}^n\to\mathbb{R}$ is said to {\em approximate}
$f$ with an error $\epsilon$, $0<\epsilon<1/2$, if
\[ |f(x)-\tilde f(x)|\le \epsilon,\quad\forall x\in\{0, 1\}^n.\]
The {\em approximate degree} of $f$, denoted by
$\appd_\epsilon(f)$ is
smallest degree of a polynomial approximating $f$ with an error
$\epsilon$. Difference choices for $\epsilon$ only result
in a constant factor difference in the approximate degrees. Thus
we omit the subscript $\epsilon$ for asymptotic estimations.

While the approximate degree of symmetric functions
has a simple characterization~\cite{NisanS92, Paturi},
it is difficult to determine in general. 
For example, the approximate degree of the two level AND-OR trees 
is still unknown. On the other hand, $\appd(f)$ is polynomially related
to the deterministic decision tree complexities $T(f)$.
Formally, $T(f)$ is defined to be the minimum integer $k$ such that
there is an ordered full binary tree $T$
of depth $k$ satisfying the following properties:
(a) each non-leaf vertex is labeled by a variable $x_i$, and 
each leaf is labeled by either $0$ or $1$ (but not both);
(b) for any $x\in\{0, 1\}^n$, the following walk leads to a leaf
labeled with $f(x)$: start from the root,
at each non-leaf vertex labeled with $x_i$, take the left
edge if $x_i=0$, and take the right edge otherwise.

\begin{theorem}[Nisan and Szegedy \cite{NisanS92}, Beals et al.
\cite{Beals:2001:QLB}]\label{thm:ns} For any Boolean function $f_n$,  there are
constants $c_1$ and $c_2$ such that
$c_1 T^{1/6}(f) \le \appd(f) \le c_2 T(f)$.
\end{theorem}

The exponent $1/6$ is not known to be optimal. The conjectured
value is $1/2$.

As observed by Buhrman, Cleve, and Wigderson~\cite{BCW98a},
a decision tree algorithm can be turned into a communication protocol
for a related problem. 
In such a protocol for $f_n \Box g_k$, one party
simulates the decision tree algorithm for $f_n$,
and initiates a sub-protocol
for computing $g_k$ whenever one input bit of $f_n$ is needed. 
The sub-protocol repeats an optimal protocol for $g_k$
for $O(\log \appd(f_n))$ times, 
ensuring that the error probability is $\le \frac{1}{3(\appd(f_n)/c_1)^6}$.
Thus the larger protocol computes 
$f_n\Box g_k$ with error probability $\le 1/3$,
and exchanges $O(R(g_k)\appd^6(f_n)\log\appd(f_n))$ bits.

\begin{proposition}[\cite{BCW98a, Beals:2001:QLB}]\label{prop:upper}
For any function $f_n\in\mathcal{F}_n$ with $\appd(f_n)=d$,
and any $g_k\in\mathcal{G}_k$,
$R(f_n\Box g_k)=O(R(g_k)d^6\log d)$.
\end{proposition}

\section{The Main Lemma}
\label{sec:mainlemma}
In this section, we prove that under some assumptions,
$Q(f_n\Box g_k)=\Omega(\appd(f_n))$.
This is shown by turning a ``witness'' for $f_n$ requiring a high
approximate degree into a ``witness'' for the hardness of
$f_n\Box g_k$.
\subsection{Witness of high approximate degree}
We now fix a function $f_n\in\mathcal{F}_n$
with  $\appd_\epsilon(f_n)=d$. For $w\in
\{0, 1\}^n$, denote by $\chi_w\in\mathcal{F}_n$
the function $\chi_w(x)=(-1)^{w\cdot x}$. Then there is no feasible solution to
the following linear system, where the unknowns are $\alpha_w$:
\begin{equation}\label{eqn:linear}
-\epsilon+f(x) \le \sum_{w:|w|<d} (-1)^{w\cdot x}\ \alpha_w \le
f(x) + \epsilon,\quad\forall x\in\{0, 1\}^n.
\end{equation}

By the duality of linear programming, there exist $q^+_x\ge
0$ and $q^-_x\ge 0$, $x\in\{0, 1\}^n$, such that
\[ \sum_{x} (q^+_x-q^-_x)\cdot\chi_w =0, \quad\forall w,\ |w|<d,\quad\textrm{and,}\]
\begin{equation}\label{eqn:dual}
\sum_{x} (q^+_x - q^-_x) f(x)+ \epsilon (q^+_x+q^-_x)
<0.\end{equation}

Define $q:\{0, 1\}^n\to \mathbb{R}$ as $q(x)=q^-_x-q^+_x$. Then
\[ q^T\chi_w = 0,\quad\textrm{and,}\quad 
\|q\|_1 < \frac{1}{\epsilon}\  q^T f.\]
Without loss of generality, assume that $q^T f=1$ (otherwise this
will hold after multiplying $q$ with an appropriate positive number).
Then $\|q\|_1<1/\epsilon$.

Since $q$ is orthogonal to all polynomials of degree less than
$d$, it has non-zero Fourier coefficients only on higher
frequencies: $q = \sum_{w:|w|\ge d} \hat q_w \chi_w$,
where $\hat q_w = \frac{1}{2^n}\sum_x q(x)\chi_w(x)$.
Since $\|q\|_1<1/\epsilon$, those Fourier coefficients must be
small:
$|\hat q_w| < \frac{1}{2^n\epsilon},\quad\forall w:|w|\ge d$.

We summarize the above discussion in the following lemma.
\begin{lemma}\label{lm:dual} Let $\epsilon\in\mathbb{R}$, $0\le \epsilon<1/2$.
For any $f\in\mathcal{F}_n$,
there exists a function $q:\{0, 1\}^n\to \mathbb{R}$ such
that: (a) $q^T f=1$, (b) $\|q\|_1< 1/\epsilon$, (c) 
$|\hat q_w|\le \frac{1}{2^n\epsilon}$,
for all $w\in\{0, 1\}^n$, and (d) $\hat q_w=0$ whenever $|w|<\appd_\epsilon(f_n)$.
\end{lemma}

\subsection{Witness of large approximate trace norm}
In order to convert a witness of high approximate degree for $f_n$
to that of large approximate trace norm for $f_n\Box g_k$, 
we need to require that $g_k$ satisfies certain property,
which we now formulate.
Let $I_A, I_B\subseteq\bool^k$.
For $b\in\bool$, we identify a probability distribution $\mu$
on $g_k^{-1}(b)\cap I_A\times I_B$ with its representation as a matrix
in $\mathbb{R}^{I_A\times I_B}$, and call it a $b$-distribution.

Recall that the {\em discrepancy} of $g_k\in\mathcal{G}_k$, denoted by $\disc(g_k)$,
is 
\[ \disc(g_k) \defeq \min_{\mu}\max_{I_A, I_B\subseteq\bool^k} \left|\sum_{(x, y)\in I_A\times I_B}
\mu(x, y) (-1)^{g_k(x, y)}\right|,\]
where $\mu$ ranges over all distributions on $\dom(g_k)$.
We define a more restricted concept of discrepancy.
\begin{definition}
\label{def:spectradis}
The {\em spectral discrepancy} of $g_k\in \mathcal{G}_k$, denoted by $\rho(g_k)$,
is the minimum $r\in\mathbb{R}$ such that there exist
$I_A, I_B\subseteq \bool^k$, and $b$-distributions
$\mu_b\in\mathbb{R}^{I_A\times I_B}$ for $g_k$, $b\in\bool$,
satisfying the following conditions:
(1) $\sqrt{|I_A|\cdot|I_B|}\cdot \|\frac{\mu_0+\mu_1}{2}\|\le 1+r$, and,
(2) $\sqrt{|I_A|\cdot|I_B|}\cdot \|\frac{\mu_0-\mu_1}{2}\|\le r$.
\end{definition}
While (1) appears contrived, it will only be used in deriving a general lower bound
on quantum communication complexity. In all of explicit applications, (1) is trivially
satisfied with $\|\frac{\mu_0+\mu_1}{2}\|= 1/\sqrt{|I_A|\cdot|I_B|}$.

Kremer~\cite{Kremer} showed that $\log(1/\disc(g_k))$ is a lower bound for the quantum
communication complexity of $g_k$ when no shared entanglement is allowed.
Linial and Shraibman~\cite{LinialS07} recently showed that 
the lower bound holds even when shared entanglement is allowed.
\begin{theorem}[Linial and Shraibman~\cite{LinialS07}] 
\label{thm:discrepancyBound}
For any $g_k\in\g_k$,
$Q(g_k)=\Omega(\log(1/\disc(g_k)))$.
\end{theorem}
Suppose that $\rho(g_k)$ is achieved with $I_A$, $I_B$ and $\mu$.
Since for any $I'_A\subseteq I_A, I'_B\subseteq I_B$,
\[|\sum_{(x, y)\in I'_A\times I'_B}\mu(x, y)(-1)^{g_k(x, y)}|
\le \sqrt{|I'_A|\cdot|I'_B|}\cdot\|\frac{\mu_0-\mu_1}{2}\|\le 
\frac{\sqrt{|I'_A|\cdot|I'_B|}}{\sqrt{|I_A|\cdot|I_B|}} \rho(g_k)
\le\rho(g_k),
\]
we have
\begin{eqnarray*}
\disc(g_k)&\le&\max_{J_A, J_B\subseteq\bool^k}|\sum_{(x, y)\in J_A\times J_B}\mu(x, y)(-1)^{g_k(x, y)}|\\
&\le& \max_{I'_A\subseteq I_A, I'_B\subseteq I_B}|\sum_{(x, y)\in I'_A\times I'_B}\mu(x, y)(-1)^{g_k(x, y)}|\\
&\le& \rho(g_k).
\end{eqnarray*}
It follows from Theorem~\ref{thm:discrepancyBound},
\begin{proposition}\label{prop:spectrad}
For any $g_k\in\mathcal{G}_k$, $Q(g_k)=\Omega(\log \frac{1}{\rho(g_k)})$.
\end{proposition}

With the concept of spectral discrepancy, we are now ready to state and prove our Main Lemma.
\begin{lemma}[Main Lemma]\label{lm:main} Let $n, k\ge1$ be integers,
$g_k\in\mathcal{G}_k$, and $f_n\in\mathcal{F}_n$.
If $\rho(g_k)\le \frac{\appd(f_n)}{2en}$, then
$Q(f_n\Box g_k) = \Omega(\appd(f_n))$.
\end{lemma}

\begin{proof} Let $d\defeq \appd(f_n)$, and $F\defeq f_n\Box g_k$.
Suppose $\rho\defeq\rho(g_k)$ is achieved
with $I_A, I_B\subseteq \{0, 1\}^k$, and $\mu_b$, $b\in\{0, 1\}$.
Denote $K_A\defeq |I_A|$, $K_B\defeq |I_B|$. 
Let $F_1$ be the restriction of $f_n\Box g_k$ on 
$(I_A\times I_B)^{\otimes n}\cap\dom(F)$.
We shall prove the desired lower bound on $F_1$.
By Lemma~\ref{lm:Razborov-Yao}, 
it suffices to prove a lower bound on
$\|F_1\|_{\epsilon', \trace}$ for $\epsilon'=1/6$.
Let $q$ be the function that exists by Lemma~\ref{lm:dual} 
with respect to $f_n$ and $\epsilon=1/3$.

For a partition $\{w_1, w_2, \cdots, w_K\}$ of $[n]$, and matrices
$A_1, A_2, \cdots, A_k\in K_A\times K_B$,
denote by $\bigotimes_{k=1}^K A_k^{w_k} \in (\mathbb{R}^{K_A\times K_B})^{\otimes n}$
the product element that has $A_k$ in the components indexed by $w_k$.
Denote by $\bar w$ the complement of $w$.
Define $h\in(\mathbb{R}^{K_A\times K_B})^{\otimes n}$ as follows
\begin{equation}\label{eqn:hdef} h\defeq \sum_{z\in\bool^n} q(z)\cdot 
\bigotimes_{i=1}^n\ \mu_{z_i}^{\otimes\{i\}}.
\end{equation}

For a matrix $A=[A_{ij}]$, denote by $\|A\|_1\defeq \sum_{i, j}|A_{ij}|$.
Then $\|\mu_0\|_1=\|\mu_1\|_1=1$, and for any $z\in\bool^n$, 
\[\|\bigotimes_{i=1}^n\ \mu_{z_i}^{\otimes\{i\}}\|_1= \Pi_{i=1}^n\|\mu_{z_i}\|_1=1.\]
Since for a different $z$, the set of the non-zero entries in $\bigotimes_{i=1}^n\ \mu_{z_i}^{\otimes\{i\}}$
is disjoint, 
\[\|h\|_1=\sum_{z\in\bool^n} |q(z)| \|\bigotimes_{i=1}^n\ \mu_{z_i}^{\otimes\{i\}}\|_1=\|q\|_1\le 
1/\epsilon.\]

Note that $\trace((\bigotimes_{i=1}^n\ \mu_{z_i}^{\otimes\{i\}})^T F)=f(z_1, z_2, \cdots, z_n)$.
Thus 
\[ \trace(h^T F)=q^T f_n = 1.\]

Fix an $\tilde F\in(\mathbb{R}^{K_A\times K_B})^{\otimes n}$ with
$|F_1(x, y)-\tilde F(x, y)|\le \epsilon'$, $\forall (x, y)\in \dom(F_1)$.
Then, \[|\trace(h^T\tilde F)|=\left|\sum_{(x, y)\in\dom(F_1)}h(x, y)\tilde F(x, y)\right|
\ge \left|\sum_{(x, y)\in \dom(F)} h(x, y)F(x, y)\right|
 - \epsilon'\|h\|_1
\ge 1-\epsilon'/\epsilon\ge 1/2.\]

Therefore,
\begin{equation}\label{eqn:traceineq}
\|\tilde F\|_\trace \ge \frac{|\trace(h^T \tilde F)|}{\|h\|} 
\ge \frac{1}{2\|h\|}.
\end{equation}

Hence we need only to prove that $\|h\|$ is very small.
To this end we first express $h$ using the Fourier representation
of $q$:
\begin{eqnarray}
h&=&\sum_{z\in\{0, 1\}^n} \sum_{w:|w|\ge d} \hat q_w (-1)^{w\cdot
z} \cdot \bigotimes_{i=1}^n
\mu_{z_i}^{\otimes\{i\}} \nonumber\\
&=&\sum_{w:|w|\ge d} \hat q_w\cdot \sum_{z\in\{0, 1\}^n}
(-1)^{w\cdot z} \cdot \bigotimes_{i=1}^n
\mu_{z_i}^{\otimes\{i\}}\nonumber\\
&=&\sum_{w:|w|\ge d}\hat q_w\cdot ((\mu_0 + \mu_1 )^{\otimes
\bar w})\otimes ((\mu_0 - \mu_1)^{\otimes w}).\label{eqn:h}
\end{eqnarray}
Using $\hat q_w\le 1/{\epsilon 2^n}$,
\comments{
\begin{eqnarray}
\|h\| &\le& \sum_{w:|w|\ge d} |\hat q_w|
\|\mu_0 + \mu_1\|^{n-|w|}\cdot \|\mu_0 - \mu_1\|^{|w|}\nonumber\\
& \le & \frac{1}{\epsilon} \sum_{\ell, \ell\ge d} {n\choose \ell}\cdot
\|\frac{\mu_0 + \mu_1}{2}\|^{n-\ell}\cdot \|\frac{\mu_0 -
\mu_1}{2}\|^{\ell}. \label{eqn:bound-h}
\end{eqnarray}
}
\begin{equation}
\|h\| \le \sum_{w:|w|\ge d} |\hat q_w|
\|\mu_0 + \mu_1\|^{n-|w|}\cdot \|\mu_0 - \mu_1\|^{|w|}
 \le  \frac{1}{\epsilon} \sum_{\ell, \ell\ge d} {n\choose \ell}\cdot
\|\frac{\mu_0 + \mu_1}{2}\|^{n-|w|}\cdot \|\frac{\mu_0 -
\mu_1}{2}\|^{|w|}\ _. \label{eqn:bound-h}
\end{equation}

By the choice of $\mu_0$ and $\mu_1$,
$\|\frac{\mu_0+\mu_1}{2}\|\le\frac{1+\rho}{\sqrt{K_A K_B}}$,
and $\|\frac{\mu_0-\mu_1}{2}\|\le\frac{\rho}{\sqrt{K_A K_B}}$.
Thus 
\begin{equation}\label{eqn:bound-r}
\|h\| \le \frac{(1+\rho)^n}{\epsilon (K_AK_B)^{n/2}} \sum_{\ell: \ell\ge d} 
{n\choose \ell} \rho^\ell.
\end{equation}

If $\rho\le \frac{d}{2en}$, using ${n\choose l} \le
(\frac{en}{l})^l$, and $(1+\rho)^n\le e^{\rho n}$,
we have 
\comments{
\begin{eqnarray}
\label{eqn:bound-final}
\|h\|&\le& \frac{e^{\rho n}}{\epsilon (K_AK_B)^{n/2}} 
\sum_{\ell\ge d}
\left(\frac{en\rho}{\ell}\right)^\ell\nonumber\\
 &\le& \frac{e^{\rho n}}{\epsilon (K_AK_B)^{n/2}} \sum_{\ell\ge d} \left(\frac{d}{2\ell}\right)^\ell\nonumber\\
&\le& \frac{2}{\epsilon (K_AK_B)^{n/2}} e^{-.5 d}.
\end{eqnarray}
}
\begin{equation}
\label{eqn:bound-final}
\|h\|\le \frac{e^{\rho n}}{\epsilon (K_AK_B)^{n/2}} 
\sum_{\ell\ge d}
\left(\frac{en\rho}{\ell}\right)^\ell
 \le \frac{e^{\rho n}}{\epsilon (K_AK_B)^{n/2}} \sum_{\ell\ge d} \left(\frac{d}{2\ell}\right)^\ell\nonumber
\le \frac{2}{\epsilon (K_AK_B)^{n/2}} e^{-.5 d}.
\end{equation}
Together with Equation~\ref{eqn:traceineq}, this implies
$ \|\tilde F\|_{\trace} \ge \frac{\epsilon}{4}\cdot (K_AK_B)^{n/2} \cdot e^{.5d}$.
Thus $\|F_1\|_{1/6, \trace} \ge \frac{1}{24}\cdot(K_AK_B)^{n/2}\cdot e^{.5d}$.
Plugging this inequality to the Razborov-Yao Lemma, we have
$Q(F)\ge Q(F_1)=\Omega(d)$.
\end{proof}

\section{Applications}
We now apply the Main Lemma to derive two quantum lower bounds.
The first deals with those $g_k$ that have polynomially 
related quantum and randomized communication complexities.
As a concrete example we consider $g_k$ being the \textsc{Inner Product}
function.
The second result shows that without this knowledge on $g_k$, 
we may still able to obtain strong quantum lower bounds.
This is done through a ``hardness amplification'' technique
that makes use of the self-similarity of the function considered.
We demonstrate this technique by proving Theorem~\ref{thm:symmetric}.

\subsection{Composition with hard $g_k$}
We now restate Theorem~\ref{res:hardg:informal}
rigorously.
\begin{theorem}\label{res:hardg}
Let $n, k\ge1$ be integers and $g_k\in\mathcal{G}_k$. 
If $Q(g_k)$ and $R(g_k)$ are polynomially related, 
so is $Q(f_n\Box g_k)$ and $R(f_n\Box g_k)$ for any $f_n\in\mathcal{F}_n$
and for $\rho(g_k)\le \frac{1}{2en}$. 
\end{theorem}

\begin{proof} If $f_n$ or $g_k$ is a constant function,
$Q(f_n\Box g_k)=R(f_n\Box g_k)=0$, hence the statement holds.
Otherwise, one can fix the value of all but one input block
so that $f_n\Box g_k$ computes $g_k$ on the remaining
block. Thus $Q(f_n\Box g_k)\ge Q(g_k)$. By Main Lemma,
under the assumption that $\rho(g_k)\le \frac{1}{2en}$,
$Q(f_n\Box g_k)=\Omega(\appd(f_n))$. Thus
$Q(f_n\Box g_k)=\Omega(\appd(f_n))+Q(g_k))$.
On the other hand $R(f_n\Box g_k)=O(R(g_k)\appd^6(f_n)\log\appd(f_n))$,
by Proposition~\ref{prop:upper}.
Thus, under the assumption that $R(g_k)$ and $Q(g_k)$
are polynomially related, so are $Q(f_n\Box g_k)$ and $R(f_n\Box g_k)$.
\end{proof}

Similarly, the same statement holds
with $R(f_n\Box g_k)$ and $R(g_k)$ replaced by
$D(f_n\Box g_k)$ and $D(g_k)$, respectively.
Estimating $\rho(g_k)$ is unfortunately difficult
in general. However, if we can show
$\rho(g_k)=\exp(-\Omega(k^c))$ for some constant $c$,
it implies $R(g_k)$ and $Q(g_k)$ are polynomially related, 
by Proposition~\ref{prop:spectrad}.
Thus $Q(f_n\Box g_k)$ and $R(f_n\Box g_k)$ are
polynomially related for $k\ge \log_2^{1/c}(2en)$.

We now prove Corollary~\ref{col:ip}. 
\begin{proofof}{Corollary~\ref{col:ip}}
We need only to consider the case that $f_n$
is not a constant function.
Then $Q(f_n\Box g_k)=\Omega(\IP_k)$.
It is known that $Q(\IP_k)=\Omega(k)$~\cite{Cleve+}.
Thus $Q(f_n\Box g_k)=\Omega(k)$.
Let $K\defeq 2^k$, $I_A\defeq\bool^k-\{0^k\}$,
and $I_B\defeq\bool^k$. For $b\in\{0, 1\}$,
let $\mu_b$ be the uniform distribution on
$\{(x, y) : \IP(x, y)=b,\ x\ne 0\}$.
Then 
\[\|\frac{\mu_0+\mu_1}{2}\|=1/\sqrt{K(K-1)},\quad\textrm{and,}\quad
\|\frac{\mu_0-\mu_1}{2}\|=1/((K-1)\sqrt{K}).\]
Thus $\rho(\IP_k)\le 1/\sqrt{K-1}$.
When $k\ge 2\log_2 n+5 >\log_2(4e^2n^2+1)$,
we have $\rho(\IP_k)\le1/2en\le \appd(f_n)/(2en)$.
By Main Lemma~\ref{lm:main}, this implies $Q(f_n\Box \IP_k)=\Omega(\appd(f_n))$.
Therefore, $Q(f_n\Box \IP_k)=\Omega(k+\appd(f_n))$.

On the other hand, $D(f_n\Box \IP_k)\le k T(f_n)$, which is $O(k\appd^6(f_n))$
by Theorem~\ref{thm:ns}.
Thus $D(f_n\Box \IP_k)=O(Q^7(f_n\Box \IP_k)$.
\end{proofof}

We remark that since for a random $g_k$, $\rho(g_k)=\exp(-\Omega(k))$,
the above corollary holds for most $g_k$ up to a constant additive difference
in the bound for $k$.

\subsection{Composition with \textsc{Set Disjointness}}
In this section we prove Theorem~\ref{thm:symmetric}.
We introduce some notions following~\cite{Razborov}.
For an integer $k\ge1$, let
$[k]\defeq\{1, 2, \cdots, n\}$. For an integer $p$, $0\le p\le k$,
denote by $[k]^p$ the set of $p$-element subsets of $[k]$.
For integers $s$ and $p$ with $0\le s\le p\le k/2$,
denote by $J_{k, p, s}\in\{0, 1\}^{[k]^p\times [k]^p}$
the indicator function for $|x\cap y|=s$. That is,
for any $(x, y)\in[k]^p\times[k]^p$,
  \[(J_{k, p, s})_{x, y}\defeq \left\{
\begin{array}{cc}
1 &\textrm{if $|x\cap y| = s$}, \\0& \textrm{otherwise.}
\end{array}\right.\] 
The spectrum of these combinatorial
matrices are described by \emph{Hahn polynomials}
\cite{Delsarte78}. We will use a formula given by Knuth
\cite{Knuth91}.
 \begin{proposition}[Knuth]\label{prop:Knuth}
Let $p\le k/2$. Then the matrices $J_{k, p, s}$, $0\le s\le p$, 
share the same eigenspaces $E_0$, $E_1$,
 $\ldots$, $E_p$, and the eigenvalue 
 corresponding to the eigenspace $E_t$, $0\le t\le p$, is given by
 \begin{equation}\label{eqn:knuth}
 \sum_{i=\max\{0, s+t-p\}}^{\min\{s,t\}} (-1)^{t-i} {t\choose i}
 {p-i \choose s-i} {k-p-t+i \choose p-s-t+i}.
 \end{equation}
 \end{proposition}
We actually need only to consider $s\in\bool$. Effectively, we 
are restricting $\disj_k$ on $\{(u, v) : u, v\in[k]^p,\ |u\cap v|\le 1\}$.
Denote this restriction by $\disjP_k$.

\begin{lemma}\label{lm:disj} Let $n, k\ge1$ be integers,
$f_n\in\mathcal{F}_n$, and $k\ge \frac{6en}{\appd(f_n)}$.
Then $Q(f_n\Box \disjP_k) = \Omega(\appd(f_n))$.
\end{lemma}
\begin{proof}
Let $p\defeq k/3$ and $ M\defeq{k\choose p}$.
Let $w_s\defeq \left|(\disjP_k)^{-1}(s)\right|$, $s\in\{0, 1\}$.
That is,
\[w_0={k\choose p}{k-p \choose p} = M{k-p \choose p},\quad\textrm{and,}
\quad
w_1={k\choose p}{p\choose 1}{k-p \choose p-1} = M{p\choose 1}{k-p \choose p-1}.\]
Let $\mu_s$, $s\in\{0, 1\}$, be the distribution matrix
for the uniform distribution
on the $s$-inputs of $\disjP_k$. That is,
\[
 \mu_0 \defeq \frac{1}{w_0} J_{k, p, 0}, \quad \textrm{and,}\quad
  \mu_1 \defeq \frac{1}{w_1} J_{k, p, 1}.
\]
By Proposition~\ref{prop:Knuth},
$\mu_0$ and $\mu_1$ have the same eigenspaces.
Furthermore, if $\lambda_{s, t}$, $s\in\bool$ and $0\le t\le p$,
is the eigenvalue of $\mu_s$ for the eigenspace $E_t$,
\begin{equation}
\lambda_{s, t} = \frac{1}{w_s} \sum_{i=\max\{0,
s+t-p\}}^{\min\{s, t\}} (-1)^{t-i} {t\choose i}
 {p-i \choose s-i} {k-p-t+i \choose p-s-t+i},
\end{equation}
and 
\begin{equation}\label{eqn:disj:discrepancy}\|\mu_0-\mu_1\|
= \max_{t: 0\le t\le p} |\lambda_{0, t} - \lambda_{1, t}|.\end{equation} 
After simplification,
\[\lambda_{0, t} = \frac{(-1)^t}{M} \frac{{k-p-t \choose p-t}}{{k-p \choose p}},\quad\textrm{and,}\quad
\lambda_{1, t} = \frac{(-1)^t}{M} \left(\frac{{k-p-t \choose
p-1-t}}{{k-p \choose p-1}} -\frac{t{k-p-t+1 \choose
p-1-t+1}}{p{k-p \choose p-1}}\right) .
\]
Since $\lambda_{0, 0} = \lambda_{1, 0} = 1$, we only need to bound
$\max_{t} |\lambda_{0, t} - \lambda_{1, t}|$ for $t\ge 1$. From Proposition~\ref{prop:Knuth},
 \begin{eqnarray}
\lambda_{0, t} - \lambda_{1, t}&=& \frac{(-1)^t}{M} \frac{{k-p-t
\choose p-t}}{{k-p \choose p}}(1- \frac{p-t}{p} +
\frac{t(k-p-t+1)}{p^2}) \nonumber\\
 &=& (-1)^t\frac{1}{M} \frac{{k-p-t
\choose p-t}}{{k-p \choose p}}\frac{t(k-t+1)}{p^2}.\nonumber
 \end{eqnarray}
With $k=3p$,
 \[ \frac{t{k-p-t \choose p-t}}{{k-p \choose p}} =
 \frac{t\cdot p\cdot(p-1)\ldots(p-t+1)}{(k-p)\cdot(k-p-1)\ldots(k-p-t+1)} 
  \le  (\frac{p}{k-p})^t\cdot  t 
 = \frac{1}{2^t} t \le \frac{1}{2}.\]
Hence
 \begin{equation}\label{eqn:disj:eigen}
|\lambda_{0, t} - \lambda_{1, t}|\le \frac{1}{2}\cdot
 \frac{k-t+1}{Mp^2}\le \frac{1}{2}\cdot \frac{k}{M(\frac{k}{3})^2} =
\frac{6}{Mk}.
 \end{equation}
Therefore, $M\|\frac{\mu_0 - \mu_1}{2}\| \le \frac{3}{k}$.
Since $\frac{\mu_0+\mu_1}{2}$ is doubly stochastic,
$\|\frac{\mu_0+\mu_1}{2}\|=1$.
Thus we have
\begin{equation}\label{eqn:rhoDISJP}
\rho(g_k)\le 3/k.
\end{equation}
Therefore, when $k\ge 6en/d$, we have $\rho(g_k)\le d/(2en)$.
By Main Lemma~\ref{lm:main}, this implies
$Q(f_n\Box \disjP_k) = \Omega(\appd(f_n))$.
\end{proof}

Let $f_n\in\mathcal{F}_n$ be a symmetric function. 
Following~\cite{Razborov}, define
\[\ell_0(f_n)\defeq\max\{ m : 1\le m\le n/2,\ f_n(1^m0^{n-m})\ne
f_n(1^{m-1}0^{n-m+1})\}\cup\{0\},\]
and
\[\ell_1(f_n)\defeq\max\{n-m: n/2\le m< n,\ f_n(1^m0^{n-m})\ne
f_n(1^{m+1}0^{n-m-1})\}\cup\{0\}.\]
We will use the following
result in proving quantum lower bounds on $f_n\Box\wedge$.

\begin{theorem}[Paturi~\cite{Paturi}]\label{thm:Paturi} Let $f_n\in\mathcal{F}_n$
be symmetric. Then for some universal
constant $c$, $\appd(f_n)\ge c\sqrt{n(\ell_0(f_n)+\ell_1(f_n))}$.
\end{theorem}

\begin{theorem}\label{thm:qlower} For any symmetric $f_n\in\mathcal{F}_n$,
$Q(f_n\Box \wedge)=\Omega( n^{1/3}\ell_0^{2/3}(f_n) + \ell_1(f_n))$.
\end{theorem}
The lower bound is weaker than Razborov's, which is
\begin{equation}\label{eqn_Rz}
Q((f_n\Box\wedge)=\Omega(\sqrt{n\ell_0(f_n)}+ \ell_1(f_n)).
\end{equation}
In the following proof, we first show that
$Q(f_n\Box\wedge)=\Omega(n^{1/3}\ell_0^{2/3}(f_n))$, then we show $Q(f_n\Box\wedge)=\Omega(\ell_1(f_n))$.
In both parts of the proof, we reduce an instance of $f_{n'}\Box \disjP_k$ to $f_n\Box\wedge)$
for some appropriate function $f_{n'}$ and $k$.
\begin{proofof}{Theorem~\ref{thm:qlower}}
Let $c$ be the constant in Theorem~\ref{thm:Paturi}, 
$\beta\defeq\min\{ \sqrt[3]{2}, \left(\frac{c}{12e}\right)^{2/3}\}$,
and $\alpha\defeq (\beta/2)^{3/2}$.

Consider the case that $\ell_0\defeq\ell_0(f_n)\le \alpha n$.
Let $n'\defeq \beta n^{2/3}\ell_0^{1/3}$,
and $f_{n'}\in\mathcal{F}_{n'}$ be such that $f_{n'}(x)=f_n(x0^{n-n'})$,
$\forall x\in\bool^{n'}$.
By direct inspection, $n'\le n$, thus $f_{n'}$ is well-defined.
Since
\[ f_{n'}(1^{\ell_0-1}0^{n'-\ell_0+1})=f_n(1^{\ell_0-1}0^{n-1\ell_0+1})
\ne f_n(1^{\ell_0}0^{n-\ell_0})=f_{n'}(1^{\ell_0}0^{n'-\ell_0}),\]
and by direct inspection, $\ell_0\le n'/2$, 
we have $\ell_0(f_{n'})\ge \ell_0$.
By Theorem~\ref{thm:Paturi}, 
\[ \appd(f_{n'})\ge c\sqrt{n' (\ell_0(f_{n'})+
\ell_1(f_{n'}))}\ge c\sqrt{n'\ell_0}.\]
Set $k\defeq\lceil \frac{6en'}{\appd(f_{n'})}\rceil$. 
By Lemma~\ref{lm:disj},
$Q(f_{n'}\Box \disjP_k)=\Omega(\appd(f_{n'}))=\Omega(n^{1/3}\ell_0^{2/3})$. 
Note that 
\[ n'k \le \beta n^{2/3}\ell_0^{1/3}\cdot 
\frac{12e\sqrt{\beta}}{c}
\left(\frac{n}{\ell_0}\right)^{1/3}=
\beta^{3/2}\frac{12e}{c}n\le n.\]
Therefore, $\forall(x, y)\in\dom(f_{n'}\Box\disjP_k)$,
we have $(f_{n'}\Box \disjP_k)(x, y)=(f_n\Box\wedge)(x0^{n-n'k}, y0^{n-n'k})$.
Thus $Q(f_n\Box\wedge)\ge Q(f_{n'}\Box \disjP_k)=\Omega(n^{1/3}\ell_0^{2/3})$.

Now consider the case that $\alpha n< \ell_0\le n/2$.
Set $k\defeq \lceil\frac{6\sqrt{2}e}{c}\rceil$, and
$n'\defeq\min \{\frac{n-\ell_0+1}{2k-1}, \ell_0-1\}$.
Then $n'=\Theta(n)=\Theta(\ell_0)$.
Define $f_{n'}\in\mathcal{F}_{2n'}$ as follows:
\[ f_{n'}(x) = f_n(x1^{\ell_0-1-n'}0^{n-2n'-(\ell_0-1-n')}),
\quad\forall x\in\bool^{2n'}.\]
By direct inspection, $f_{n'}$ is well-defined.
Then 
\[ f_{n'}(1^{n'}0^{n'})=f_n(1^{\ell_0-1}0^{n-\ell_0+1})
\ne f_n(1^{\ell_0}0^{n-\ell_0})=f_{n'}(1^{n'+1}0^{n'-1}).\]
Therefore, $\ell_1(f_{n'})= n'$, and $\appd(f_{n'})\ge \sqrt{2}c n'$, 
by Theorem~\ref{thm:Paturi}. 
By direct inspection, $k\ge \frac{6e(2n')}{\appd(f_{n'})}$, thus
$Q(f_{n'}\Box \disjP_k)=\Omega(\appd(f_{n'}))=\Omega(n')$.
Note that for all $(x, y)\in \dom(f_{n'}\Box \disjP_k)$,
\[(f_{n'}\Box \disjP_k)(x, y)=(f_n\Box \wedge)(x 1^{\ell_0-1-n'}0^{n-(\ell_0-1-n')-2kn'},
y 1^{\ell_0-1-n'}0^{n-(\ell_0-1-n')-2kn'}).\]
By direct inspection, the number of $0$'s and $1$'s padded in the above equation
is non-negative.
Thus 
\[Q(f_n\Box\wedge)=\Omega(Q(f_{n'}\Box \disjP_k)=\Omega(n')=\Omega(\ell_0)=
\Omega(n^{1/3}\ell_0^{2/3}).\]

We use a similar reduction
to prove $Q(f_n\Box\wedge)=\Omega(\ell_1)$.
Let $k$ be the same as above.
Set $n'\defeq \lfloor \frac{\ell_1}{2k-1}\rfloor$,
and define $f_{n'}\in\mathcal{F}_{2n'}$ as follows
\[ f_{n'}(x)=f_n(x1^{n-\ell_1-n'} 0^{n-2n'-(n-\ell_1-n')})\quad\forall x\in\bool^{2n'}.\]
By direct inspection, the numbers of 
padded $0$'s and $1$'s are non-negative, thus
$f_{n'}$ is well-defined.
Since
\[ f_{n'}(1^{n'}0^{n'})=f_n(1^{n-\ell_1}0^{\ell_1})\ne
f_n(1^{n-\ell_1+1}0^{n-\ell_1-1})=f_{n'}(1^{n'+1}0^{n'-1}),\]
we have $\ell_1(f_{n'})=n'$. Thus 
$\appd(f_{n'})\ge \sqrt{2}cn'$
by Theorem~\ref{thm:Paturi}, and 
$Q(f_{n'}\Box\disjP_k)=\Omega(\appd(f_{n'}))=\Omega(\ell_1)$
by Lemma~\ref{lm:disj}.
For all $(x, y)\in\dom(f_{n'}\Box\disjP_k)$,
\[(f_{n'}\Box \disjP_k)(x, y) = (f_n\Box\wedge)(x 1^{n-\ell_1-n'}0^{n-2kn'-(n-\ell_1-n')},
y 1^{n-\ell_1-n'}0^{n-2kn'-(n-\ell_1-n')}).\]
By direct inspection again, 
the numbers of the padded digits in the above are non-negative.
Thus $Q(f_n\Box\wedge)\ge Q(f_{n'}\disjP_k)=\Omega(\ell_1)$.
\end{proofof}

Next, we establish a classical upper bound
on the randomized complexity of symmetric
predicates. 

\begin{proposition}\label{prop:rupper} Let $f_n\in\mathcal{F}_n$ be symmetric
with $\ell_0(f_n)=0$. Then
\[R(f_n\Box\wedge)=O(\ell_1\log^2\ell_1\log\log\ell_1).\]
\end{proposition}

Theorem~\ref{thm:symmetric} follows from Theorem~\ref{thm:qlower}
and Proposition~\ref{prop:rupper}: if $\ell_0(f_n)\ge 1$, $Q(f_n\Box\wedge)=\Omega(n^{1/3})
=\Omega(D^{1/3}(f_n\Box\wedge)=\Omega(R^{1/3}(f_n\Box\wedge))$.
Otherwise, $Q(f_n\Box\wedge)=\Omega(\ell_1(f))=\Omega(R^{1/2}(f_n\Box\wedge))$.
Similarly, Proposition~\ref{prop:quadGap} follows from Proposition~\ref{prop:rupper}
and Razborov's lower bound Equation~\ref{eqn_Rz}.

To prove Proposition~\ref{prop:rupper}, we use the following result from Huang et al.~\cite{Hamming}.
Let $n$ and $d$ be integers with $0\le d\le n$. The \textsc{Hamming Distance
Problem} $\ham_{n, d}$ is defined as
\begin{equation*} \ham_{n, d}(x, y)=\left\{\begin{array}{cc} 1&\textrm{$|x\oplus y|\ge d$,}\\
0&\textrm{otherwise.}\end{array}\right.\end{equation*}

\begin{theorem}[Huang et al.~\cite{Hamming}] \label{thm:Hamming}
There is randomized protocol for $\ham_{n, d}$ 
that exchanges $O(d\log d)$ bits and errs with probability $\le 1/3$.
\end{theorem}

\begin{proofof}{Proposition~\ref{prop:rupper}}
Without loss of generality, assume $f_n(1^m0^{n-m})=0$ for all $m$, $0\le m\le n-\ell_1$.
The following randomized protocol computes $f_n\Box \wedge$
with $O(\ell_1\log^2\ell_1\log\log\ell_1)$ bits of communication.
Fix an input $(x, y)$, and let $z_A\defeq n-|x|$ and $z_B\defeq n-|y|$.
Alice and Bob first check if $z_A\ge \ell_1$ or $z_B\ge \ell_1$.
If yes, they output $0$ and terminate the protocol. Otherwise,
Alice sends $z_A$ to Bob using $\lceil\log_2(\ell_1-1)\rceil$ bits, and they 
compute $\delta\defeq|x\oplus y|$. Knowing $z_A$ and $\delta$, Bob
is able to compute $f(|x\cap y|)=f((|x|+|y|-|x\oplus y|)/2)$.
Note that $\Delta\defeq 2(\ell_1-1)\ge \delta\ge 0$.
Thus Alice and Bob can perform a binary search to determine $\delta$
with $\log_2(\Delta+1)$ sub-protocols for the \textsc{Hamming Distance
Problem}.
For each candidate value $d$ of $\delta$, they repeat
the randomized protocol in Theorem~\ref{thm:Hamming}
for $\ham_{n, d}$ for $\Theta(\log \log \Delta)$ times so that
the error probability is $\le\frac{1}{3(\log_2\Delta  +1)}$. 
Thus the total number of bits exchanged is
$O(\Delta\log^2\Delta\log\log \Delta)=
O(\ell_1\log^2\ell_1\log\log\ell_1)$,
and the error probability of the complete protocol is $\le 1/3$.
\end{proofof}

\begin{remark}
While both Razborov's proof and the above use the spectrum
decompositions of the matrix $J_{k, p, s}$, we emphasize
their difference: we only need to analyze
$\|\frac{\mu_0-\mu_1}{2}\|$, which corresponds to $s=0, 1$.
In contrast, Razborov's proof needs much more details
of the spectrum decompositions, in particular,
it needs to consider $s=0, 1, \cdots, \Theta(n)$.

Theorem~\ref{thm:symmetric} implies $Q(\disj_n)=\Omega(n^{1/3})$.
Note that our estimate (Equation~\ref{eqn:rhoDISJP}) gives 
$\rho(\disj_k)=O(1/k)$. Thus by Proposition~\ref{prop:spectrad},
this only gives a very weak lower bound $Q(\disj_n)=\Omega(\log n)$.
Surprisingly, this weak bound can be amplified to $\Omega(n^{1/3})$
through the dual formulation of the approximate degree (Lemma~\ref{lm:dual}).
Finding more examples of such ``hardness amplification'' would
be very interesting.
\end{remark}

\section{Open problems and discussions}
While the block-composed functions we focus on are restricted
to have identical $g_k$ in each block, and $g_k$ has balanced
input size on Alice and Bob's side, our technique can be extended
straightforwardly to deal with non-identical, and general 
building block functions. Pushing this approach to its limit
in resolving the Log-Equivalence Conjecture is an interesting
direction. 

A specific problem is to minimize the technical assumption
on the block-size in the Main Lemma --- for some $g_k$, this can be
accomplished by using the result of Sherstov~\cite{Sherstov}, 
which we will describe below in more details.
Another specific problem is to prove the Log-Equivalence Conjecture
for $f_n\Box\wedge$, for an arbitrary $f_n$.

In an independent work, Sherstov~\cite{Sherstov} also 
derived Lemma~\ref{lm:dual}, and used it to prove strong 
quantum lower bounds on
what he called ``pattern matrices''. In our notation,
he considered functions $f_n\Box g^0_k$, where 
$f_n\in\mathcal{F}_n$ and $g^0_k:\{0, 1\}^k\times
([k]\times\bool)\to\bool$ is fixed with
$g^0_k(x, (i, b))\defeq x_i+b$. His main result is, $Q(f_n\Box g^0_k)
=\Omega(\appd(f_n))$ for any $f_n$. The proof 
also starts with the dual characterization of $\appd(f_n)$, 
constructs $q$ via Lemma~\ref{lm:dual},
then constructs a witness matrix $h$ (or $K$ in~\cite{Sherstov})
for the high trace norm of any matrix approximating $f_n\Box g_k^0$.
His construction of $h$ can be expressed in the same equation 
(Eqn.~\ref{eqn:hdef}) as ours with carefully chosen $\mu_0$ and $\mu_1$
for $g_k^0$.

The main technical difference takes place after Eqn.~(\ref{eqn:h}).
With the fixed $g_k^0$, the constructed $h$ has the 
nice property that the left and right eigenvectors of $(\mu_0+\mu_1)^{\otimes
\bar w}\otimes(\mu_0-\mu_1)^{\otimes w}$ are in orthogonal subspaces,
due to the fact that 
\begin{equation}\label{eqn:ortho}
(\mu_0+\mu_1)^T(\mu_0-\mu_1)=0,\quad\textrm{and,}\quad
(\mu_0-\mu_1)^T(\mu_0+\mu_1)=0.
\end{equation} 
Thus, he was able to avoid the use of the triangle inequality in Eqn.~(\ref{eqn:bound-h})
and replace the summation by the maximum. This sharper bound
moderates the requirement on $k$, and results in an alternative
proof for Razborov's lower bound with the same asymptotic parameters and
without using Hahn polynomials at all. 
In particular, he proved that
$Q(f_n\Box \disj_k)=\Omega(\appd(f_n))$ for any $f_n$ and any $k\ge4$.
This is a significantly stronger result than our requirement
that $k\ge \frac{6en}{\appd(f_n)}$ (Lemma~\ref{lm:disj}) when
$\appd(f_n)$ is much smaller than $n$.
On the other hand, for a general $g_k$, the best bound on $Q(f_n\Box g_k)$
provable through this method (i.e., using pairs of
$\mu_0$ and $\mu_1$ satisfying the orthogonality condition
(\ref{eqn:ortho})), is not necessarily
stronger than that in Main Lemma. This is because the orthogonality condition 
restricts the choice of $\mu_0$ and $\mu_1$ to smaller domains. 

\section{Acknowledgments} We thank Jianxin Chen, Sasha Razborov, Sasha Sherstov,
and Zhiqiang Zhang for useful discussions. We also thank Sasha Sherstov for sending
us his manuscript~\cite{Sherstov}.


\end{document}